\documentclass[sigconf, natbib=false, screen]{acmart}

\renewcommand\footnotetextcopyrightpermission[1]{} 
\acmConference[ASPDAC '21]{26th Asia and South Pacific Design Automation Conference}{January 18--21, 2021}{Tokyo, Japan}

\usepackage[english]{babel}
\usepackage{mathtools,amsthm}

\usepackage{float}
\usepackage{dblfloatfix}
\usepackage{enumitem}
\usepackage{soul}
\usepackage{xcolor}
\usepackage{booktabs}
\usepackage{csvsimple}
\usepackage{csquotes}
\usepackage[style=ieee, maxnames=3, minnames=1, date=year, doi=false,isbn=false,backend=biber, sortcites, dashed=false]{biblatex}
\usepackage{subcaption}

\usepackage{tikz}
\usetikzlibrary{arrows.meta,backgrounds,calc,chains,matrix,positioning,shapes,shapes.geometric,	shapes.arrows, decorations.pathmorphing, decorations.pathreplacing}

\usepackage{stmaryrd}
\usepackage[binary-units=true, output-decimal-marker={.}]{siunitx}
\usepackage{etoolbox}
\robustify\bfseries
\usepackage[textsize=tiny]{todonotes}

\newtheorem{example}{Example}
\makeatother

\usepackage{nicematrix}
\usepackage{nicefrac}

\usepackage{comment}

\newtheorem{corollary}{Corollary}
\newtheorem{theorem}{Theorem}

\newcommand{\qop}[1]{\ensuremath{\mathit{#1}}}
\DeclarePairedDelimiterX{\abs}[1]{\lvert}{\rvert}{\ifblank{#1}{{}\cdot{}}{#1}}
\newcommand{\CC}{C\nolinebreak\hspace{-.05em}\raisebox{.2ex}{\footnotesize\bf +}\nolinebreak\hspace{-.10em}\raisebox{.2ex}{\footnotesize\bf +}}
\usetikzlibrary{quantikz}

\makeatletter
\def\thm@space@setup{%
  \thm@preskip=.5\topsep
  \thm@postskip=\thm@preskip %
}
\makeatother
 
\hypersetup{
	pdftitle={Random Stimuli Generation for the Verification of Quantum Circuits},
	pdfauthor={Lukas Burgholzer, Richard Kueng, and Robert Wille}
} 
 
\begin{document}

\title[Random Stimuli Generation for the Verification of Quantum Circuits]{\resizebox{\linewidth}{!}{Random Stimuli Generation for the Verification of Quantum Circuits}}

\author[Lukas Burgholzer, Richard Kueng, and Robert Wille]{\vspace*{-6mm}Lukas Burgholzer$^*$\hspace{3.0em}Richard Kueng$^*$\hspace{3.0em}Robert Wille$^{*\dagger}$}
\affiliation{%
   \institution{$^*$Johannes Kepler University Linz, Austria}
}
\affiliation{%
  \institution{$^\dagger$Software Competence Center Hagenberg GmbH (SCCH), Austria}
}
\email{{lukas.burgholzer, richard.kueng, robert.wille}@jku.at}
\email{https://iic.jku.at/eda/research/quantum/}

\begin{abstract}  
\vspace*{-0.35em}
Verification of quantum circuits is essential for guaranteeing correctness of quantum algorithms
and/or quantum descriptions across various levels of abstraction. 
In this work, we show that there are promising ways to check the correctness of quantum circuits using simulative verification and random stimuli.
To this end, we investigate %
how to properly generate stimuli for efficiently checking the correctness of a quantum circuit.
More precisely, we introduce, illustrate, and analyze three schemes for quantum stimuli generation---offering a trade-off between the error detection rate (as well as the required number of stimuli) and efficiency. 
In contrast to the verification in the classical realm, we show (both, theoretically and empirically) that even if only a few \emph{randomly-chosen} stimuli (generated from the proposed schemes) are considered, high error detection rates can be achieved for quantum circuits. %
The results of these conceptual and theoretical considerations have also been empirically confirmed---with a grand total of approximately $10^6$ simulations conducted across \si{50\,000} benchmark instances.
\end{abstract}

\maketitle

\vspace*{-3mm}
\section{Introduction}
\label{sec:intro}
\vspace*{-0.35em}

Verification methods are essential for demonstrating or even proving the correctness of classical circuits.
Their goal is to confirm whether a given circuit realization conforms to its specification.
In this regard, \emph{formal verification} methods~\cite{biereSATATPGBoolean2002,drechslerAdvancedFormalVerification2004}---which aim to prove correctness with \SI{100}{\percent} certainty---are well established, but often fail due to the exponential complexity of the task itself.
In contrast, \emph{simulative verification} methods~\cite{yuanConstraintbasedVerification2006,bergeronWritingTestbenchesUsing2006,kitchenStimulusGenerationConstrained2007,wille2009smt, leDetectionHardwareTrojans2019,laeuferRFUZZCoveragedirectedFuzz2018} are typically very fast as long as only a limited number of simulations with specific stimuli are conducted to achieve a desired coverage.
In order to generate high quality stimuli (which indeed are capable of detecting errors), methods such as \emph{constraint-based random simulation}~\mbox{\cite{yuanConstraintbasedVerification2006,bergeronWritingTestbenchesUsing2006,kitchenStimulusGenerationConstrained2007,wille2009smt}}, \emph{fuzzing}~\mbox{\cite{leDetectionHardwareTrojans2019,laeuferRFUZZCoveragedirectedFuzz2018}}, etc.~are employed.

In the quantum realm, the verification of quantum circuits is essential for guaranteeing correctness of quantum algorithms and/or quantum descriptions across various levels of abstraction. %
Here, sequences of quantum operations and/or quantum gates are employed which utilize quantum mechanical effects such as \emph{superposition}, \emph{entanglement}, or \emph{interference}~\cite{nielsenQuantumComputationQuantum2010}.
This allows for promising applications in various domains such as chemistry, finance, cryptography, or machine learning. But it also requires a more complex description than in the classical realm.
Consequently, the formal verification of quantum circuits poses even more challenges than in the classical realm---which even recent advances~\cite{burgholzerAdvancedEquivalenceChecking2020,burgholzerVerifyingResultsIBM2020,duncanGraphtheoreticSimplificationQuantum2019,yamashitaFastEquivalencecheckingQuantum2010,yamashitaFastEquivalencecheckingQuantum2010, ardeshir-larijaniAutomatedEquivalenceChecking2018} can only escape to a certain extent.

This motivates the consideration of simulative verification in the quantum realm (similar to the classical realm, where this is well established).
In this regard, the simulation of quantum circuits on a classical computer hardware is key.
Although this leads to an exponential complexity in order to describe the corresponding quantum states and operations, %
powerful methods have recently been proposed to tackle this problem~\cite{guerreschiIntelQuantumSimulator2020, jonesQuESTHighPerformance2018,villalongaFlexibleHighperformanceSimulator2019, pednaultLeveragingSecondaryStorage2019,seddonQuantifyingQuantumSpeedups2020, niemannQMDDsEfficientQuantum2016, zulehnerAdvancedSimulationQuantum2019,zulehnermatrix-vector2019}.
However, while the stimuli space for classical circuits is finite (each input bit can be assigned either $0$ or $1$---yielding a total of~$2^n$ possible stimuli), the state space in the quantum realm is infinitely large (possible stimuli are elements of a $2^n$-dimensional Hilbert space). 
This raises the question on whether simulative verification of quantum circuits (on classical computers) is suitable at all and, if so, how to generate proper stimuli to efficiently check the correctness of a quantum circuit.

In this work, we show that, although the perspective of a possible infinite number of stimuli may seem rather grim at a first glance, there are promising ways to check the correctness of quantum circuits using simulative verification and random stimuli. This, however, severely depends on how the stimuli are actually generated.
In fact, we introduce, illustrate, and analyze three schemes for quantum stimuli generation offering a nice trade-off between error detection rate (as well as the required number of stimuli) and efficiency. 
In contrast to classical circuits, we show (both, theoretically and empirically) that even if only a few \emph{randomly-chosen} stimuli (generated from the proposed schemes) are considered, high error detection rates can be achieved in the quantum realm.
The results of these conceptual and theoretical considerations have also been empirically confirmed, which, to the best of our knowledge, led to the broadest empirical evaluation of %
simulative verification schemes for quantum circuits to date---with a grand total of approximately $10^6$ simulations conducted across \si{50\,000} benchmark instances.

The remainder of this paper is structured as follows:
Section~\ref{sec:motivation} provides the necessary background on classical verification, quantum circuits, and their verification.
Then, Section~\ref{sec:proposed} introduces, illustrates, and (theoretically) analyzes different stimuli generation schemes and their likeliness of detecting errors. 
The results of these conceptual and theoretical considerations are then empirically confirmed in Section~\ref{sec:results}.
Finally, Section~\ref{sec:conclusions} concludes the paper. 

\vspace*{-0.8em}
\section{Background and Motivation}
\label{sec:motivation}
\vspace*{-0.25em}

This work deals with verification of circuits---a topic which has been and currently still is heavily considered in the classical realm.
Because of this, we first briefly review the established schemes in this section. Afterwards, we provide the basics on quantum computing and quantum circuits and, based on that, eventually discuss the challenges of the verification of quantum circuits. By this, we motivate our work.

\vspace*{-0.75em}
\subsection{Verification of Classical Circuits}
\label{sec:classical}
\vspace*{-0.25em}

\sloppypar
In order to demonstrate or even prove the correctness of classical circuits, 
verification methods are applied.
They check whether a given circuit, the \emph{Design Under Verification}~(DUV), adheres to an also given \emph{Golden Specification}.
To this end, current (industrial) practice mainly applies schemes such as 
\begin{itemize}
	\item {\em simulative verification}~\cite{yuanConstraintbasedVerification2006,bergeronWritingTestbenchesUsing2006,kitchenStimulusGenerationConstrained2007,wille2009smt, leDetectionHardwareTrojans2019,laeuferRFUZZCoveragedirectedFuzz2018}, in which certain input assignments (\emph{stimuli}) are explicitly assigned to the circuit, propagated through it, and the outputs are compared to the expected values, or
	
	\item {\em formal verification}~\cite{biereSATATPGBoolean2002,drechslerAdvancedFormalVerification2004}, which considers the problem
	mathematically and proves that a circuit is correct with 100\% certainty.
\end{itemize}

Obviously, formal verification provides the best solution with respect to quality.
Corresponding methods are capable of efficiently traversing large parts of the
search space, e.g., by applying clever implications during the proof. The corresponding
techniques are, however, rather complex compared to their
simulative counterparts and, particularly for larger designs, often fail due to the exponential complexity of the task.

Simulation is much easier to implement and very fast as long as only a limited number of stimuli is applied. 
The problem obviously is the quality provided by the applied set of stimuli. An exhaustive set of stimuli 
would show correctness with 100\% certainty, but is practically intractable as this would eventually require an exponential number of stimuli to simulate. Accordingly, methods such as \emph{constraint-based random simulation}~\cite{yuanConstraintbasedVerification2006,bergeronWritingTestbenchesUsing2006,kitchenStimulusGenerationConstrained2007,wille2009smt}, \emph{fuzzing}~\cite{leDetectionHardwareTrojans2019,laeuferRFUZZCoveragedirectedFuzz2018}, etc.~are key
techniques to cope with this problem while still maintaining a high quality. 
Here, stimuli and/or data inputs are specifically generated (e.g., from constraints, mutations of randomly generated inputs, etc.) 
so that corner case scenarios and/or a broad variety of cases are triggered. In doing so, errors that might
otherwise remain undetected are more likely to be found.

However, despite substantial progress that has been made in the past, e.g., on improving the efficiency of formal methods or on stimuli generation which increases the coverage of simulative verification, verifying classical circuits remains a challenge and, hence, is subject of further research.

\subsection{Quantum Circuits}
\label{sec:quantum}

Quantum circuits promise more potential than classical circuits for many applications, but also require a more complex description.
In contrast to classical bits, the main computational unit of quantum circuits (the \emph{qubit}) cannot only be in one of the computational basis states \ket{0} or~\ket{1}, but also in a \emph{superposition} of both. That is, the state \ket{\varphi} of a qubit can be described as $\ket{\varphi} = \alpha_0 \ket{0} + \alpha_1 \ket{1}$ with \mbox{$\alpha_0,\,\alpha_1\in\mathbb{C}$} and \mbox{$\abs{\alpha_0}^2 + \abs{\alpha_1}^2 = 1$}.
More generally, the state of an $n$-qubit system is described by $2^n$ complex amplitudes $\alpha_i$---each associated to a computational basis state $\ket{i}=\ket{(i_{n-1}\dots i_0)_2}=\ket{i_{n-1}}\otimes \dots \otimes \ket{i_0}$.
It holds that $\ket{\varphi} = \sum_{i\in\{0,1\}^n} \alpha_i \ket{i}$ with \mbox{$\alpha_i\in\mathbb{C}$} and $\sum_{i\in\{0,1\}^n} \abs{\alpha_i}^2 = 1$.
Typically, those states %
are expressed as \mbox{$2^n$-dimensional} \emph{state vectors} consisting of all amplitudes,~i.e., \linebreak \mbox{$\ket{\varphi} \equiv [\alpha_0,\dots,\alpha_{2^n-1}]^\top$}.

\begin{example}\label{ex:state}
	Consider the two-qubit quantum state $\ket{\varphi}$ described by $\ket{\varphi}=\nicefrac{1}{\sqrt{2}} \ket{00} + 0 \ket{01} + 0 \ket{10} + \nicefrac{1}{\sqrt{2}} \ket{11}$.
	This is a valid quantum state since $\abs{\nicefrac{1}{\sqrt{2}}}^2 + \abs{\nicefrac{1}{\sqrt{2}}}^2 = \nicefrac{1}{2} + \nicefrac{1}{2} = 1$.
	Its state vector representation is given by $[\nicefrac{1}{\sqrt{2}},0,0,\nicefrac{1}{\sqrt{2}}]^\top$.
	Notably,~\ket{\varphi} is an example of an \emph{entangled} state where the state of one qubit inherently depends on the state of another qubit---a phenomenon unique to quantum computing.
\end{example}

A \emph{quantum circuit} manipulates the state of a quantum system. To this end, each \emph{quantum gate} of a circuit realizes a certain quantum operation.
Mathematically, these operations are represented by \mbox{$2^n\times 2^n$-dimensional}, unitary matrices\footnote{A complex matrix $U$ is unitary if $U^\dag U = U U^\dag = \mathbb{I}$, where $U^\dag$ denotes the conjugate-transpose of $U$ and $\mathbb{I}$ the identity matrix.} $U$ acting on the \linebreak $2^n$-dimensional state vector $\ket{\varphi} \equiv [\alpha_0,\dots,\alpha_{2^n-1}]^\top$.
Typically, quantum operations only act on $k<n$ qubits (predominantly $k=1$ or $k=2$) and, hence, are characterized by $2^k\times 2^k$-dimensional, unitary matrices which are extended to the full system size by tensor products with identity matrices.

\vspace{200cm}

\begin{example}\label{ex:gates}
	Popular single-qubit gates include the Pauli gates $X$, $Y$, and $Z$, the Hadamard gate $H$, as well as the the phase gate $S$. The respective matrices are:
	\[
		X=\begin{bNiceMatrix}[small]
			0 & 1 \\ 1 & 0
		\end{bNiceMatrix}\quad
		Y=\begin{bNiceMatrix}[small]
			0 & -i \\ i & 0
		\end{bNiceMatrix}\quad
		Z=\begin{bNiceMatrix}[small]
			1 & 0 \\ 0 & -1
		\end{bNiceMatrix}\quad
		H=\nicefrac{1}{\sqrt{2}}\begin{bNiceMatrix}[small]
			1 & 1 \\ 1 & -1
		\end{bNiceMatrix}\quad
		S=\begin{bNiceMatrix}[small]
			1 & 0 \\ 0 & i
		\end{bNiceMatrix}.\]
	Most multi-qubit gates are \emph{controlled} gates, where a certain \mbox{single-qubit} gate is applied to a specified \emph{target} qubit only if all designated \emph{control} qubits are \ket{1}.
	One prominent example is the two-qubit \mbox{controlled-NOT} (\qop{CNOT}), which is described by the matrix
	\[
	\qop{CNOT}(q_c,q_t) = 
	\begin{bNiceMatrix}[small]
		1 & 0 & 0 & 0 \\
		0 & 1 & 0 & 0 \\
		0 & 0 & 0 & 1 \\
		0 & 0 & 1 & 0
	\end{bNiceMatrix}.\]
\end{example}
The overall quantum circuit~$G$ (realizing a quantum algorithm) is eventually represented as
a sequence of quantum gates~$g_i$, i.e., by $G=g_0,\dots,g_{m-1}$ with $m$ being the total number of gates.
The functionality of this circuit is described by the unitary matrix \mbox{$U=U_{m-1}\cdot\hdots\cdot U_0$}, where $U_i$ is the unitary matrix corresponding to gate $g_i$.

\begin{example}\label{ex:circuits}
	Consider the quantum circuit $G=g_0g_1$ acting on two qubits (denoted $q_0$ and $q_1$) with $g_0=H(q_1)$ (i.e.,~an \qop{H} gate applied to~$q_1$) and $g_1=\mathit{CNOT}(q_1,q_0)$ (i.e.,~a \qop{CNOT} gate with control qubit~$q_1$ and target qubit~$q_0$). Then, the respective matrices $U_0$, $U_1$, and the overall system matrix $U=U_1\cdot U_0$ are given by
	\[
	U_0 = H \otimes \mathbb{I}_2 = \frac{1}{\sqrt{2}}
	\begin{bNiceMatrix}[small]
		1 & 0 & 1 & 0 \\
		0 & 1 & 0 & 1 \\
		1 & 0 & -1 & 0 \\
		0 & 1 & 0 & -1
	\end{bNiceMatrix}\quad
	U_1 =
	\begin{bNiceMatrix}[small]
		1 & 0 & 0 & 0 \\
		0 & 1 & 0 & 0 \\
		0 & 0 & 0 & 1 \\
		0 & 0 & 1 & 0
	\end{bNiceMatrix}\quad
	U = \frac{1}{\sqrt{2}}
	\begin{bNiceMatrix}[small]
		1 & 0 & 1 & 0 \\
		0 & 1 & 0 & 1 \\
		0 & 1 & 0 & -1 \\
		1 & 0 & -1 & 0
	\end{bNiceMatrix}
	.
	\]
\end{example}
For more details about quantum computing we refer to~\cite{nielsenQuantumComputationQuantum2010,watrousTheoryQuantumInformation2018}.

\subsection{Verification of Quantum Circuits}
\label{sec:quantumverification}

In the quantum realm, the verification problem can be stated in a similar fashion as for classical circuits: 
Given a circuit \linebreak \mbox{$G=g_0\dots g_{m-1}$}, it should be checked whether it adheres to an also given specification\footnote{Note that the terms \emph{Device Under Verification} and \emph{Golden Specification} are not established in the quantum realm (yet), which is why we simply use the terms \emph{circuit} and \emph{specification} in the following.}.
For the sake of this work and without loss of generality we assume in the following that the specification is given as a unitary function $U$---possibly described by a \mbox{high-level} quantum algorithm, another circuit, or further functional representations for quantum computing.

However, due to the more complex/expressive description, the formal verification of quantum circuits poses even more challenges than in the classical realm. 
Despite recent advances in the design of diverse/efficient formal verification methods~\cite{burgholzerAdvancedEquivalenceChecking2020,burgholzerVerifyingResultsIBM2020,duncanGraphtheoreticSimplificationQuantum2019,yamashitaFastEquivalencecheckingQuantum2010,yamashitaFastEquivalencecheckingQuantum2010, ardeshir-larijaniAutomatedEquivalenceChecking2018}, 
these can only escape the imminent complexity to a certain extent.
Accordingly, simulative verification might provide a promising alternative as well. 
In fact, this
has  already been considered
in theoretical quantum information, where 
(truly quantum-based) methods have been proposed (see,~e.g.,~\cite[Section~3]{watrousTheoryQuantumInformation2018} and~\cite{khatriQuantumassistedQuantumCompiling2019}). But these approaches would require an execution 
on actual quantum computing devices, whose availability and accessibility still is severely restricted.
Hence, before valuable quantum computing resources are wasted to verify a quantum circuit, efficient alternatives which can be employed prior to an actual execution on a quantum computer (using classical computing devices) are of high interest\footnote{This has similarities to the verification of classical circuits which also shall be conducted prior to an actual execution in the field.}.

\vspace{200cm}

This eventually results in the following simulative verification scheme for quantum circuits:
\begin{enumerate}
	\item Consider a set $\mathcal{S}$ of quantum states (which serve as stimuli).
	\item Pick (and prepare) a quantum state $\ket{\varphi} \in \mathcal{S}$. %
	\item Simulate (on a classical device) both $U$ and $G$ with this initial state---resulting in two states \ket{\varphi_U} and \ket{\varphi_G}, respectively.
	\item Compare the output \ket{\varphi_G} generated by the realization $G$ with the desired output \ket{\varphi_U} by computing the quantum fidelity~$\mathcal{F}$ between both states~\cite{nielsenQuantumComputationQuantum2010}\footnote{In this regard the fidelity $\mathcal{F}$ acts as a similarity measure between two states---effectively computing the squared overlap of the states' amplitudes.},~i.e., 
	\[\mathcal{F}(\ket{\varphi_U}, \ket{\varphi_G}) = \abs{\braket{\varphi_U}{\varphi_G}}^2\in[0,1].\]
	\item If $\mathcal{F}(\ket{\varphi_U}, \ket{\varphi_G})\neq 1$, the stimulus \ket{\varphi} shows the incorrect behavior of~$G$ with respect to~$U$. Accordingly, the verification failed and the process is terminated.
	\item Remove $\ket{\varphi}$ from $\mathcal{S}$. 
	\item If $|\mathcal{S}| \neq \emptyset $ (i.e.,~$\mathcal{S}$ is still non-empty) continue with Step~(2); otherwise, the simulative verification process has been completed.
\end{enumerate}

Now, the challenges of such an approach are as follows:
First, in order to simulate a quantum circuit $G=g_0\dots g_{m-1}$ starting with an initial state \ket{\varphi} on a classical device (Step~(3) from above), matrix-vector multiplications of the matrices $U_i$ (representing the circuit's gates~$g_i$) with the state vector \ket{\varphi} as well as the resulting output vectors, respectively, have to be conducted consecutively. 
\begin{example}\sloppypar
	Consider the circuit $G$ from Example~\ref{ex:circuits} and the initial state $\ket{\varphi}=\ket{00} \equiv [1,0,0,0]^\top$. Applying the gate $g_0 = H(q_1)$ to this initial state,~i.e.,~computing $U_0\ket{\varphi}$, produces a new state $\ket	{\varphi^\prime} = \nicefrac{1}{\sqrt{2}} \ket{00} + \nicefrac{1}{\sqrt{2}}\ket{10} \equiv [\nicefrac{1}{\sqrt{2}},0,\nicefrac{1}{\sqrt{2}},0]^\top$. Afterwards, applying \mbox{$g_1=\qop{CNOT}(q_1,q_0)$} to $\ket{\varphi^\prime}$,~i.e.,~computing $U_1\ket{\varphi^\prime}$, results in the final state $\ket{\varphi^{\prime\prime}}=\nicefrac{1}{\sqrt{2}} \ket{00} + \nicefrac{1}{\sqrt{2}}\ket{11}\equiv [\nicefrac{1}{\sqrt{2}},0,0,\nicefrac{1}{\sqrt{2}}]^\top$---representing the output state generated by this circuit for stimulus/input $\ket{\varphi}$.
\end{example}

\sloppypar This leads to an exponential complexity since the involved vectors and matrices have a size of $2^n$ and $2^n \times 2^n$, respectively (with~$n$ being the number of qubits).
But although this is substantially harder than for the verification of classical circuits (here, a single simulation yields only linear complexity),
rather powerful methods have been recently proposed to tackle this complexity---including methods based on highly optimized and parallel \mbox{matrix-computations}~\cite{guerreschiIntelQuantumSimulator2020, jonesQuESTHighPerformance2018},
tensor networks~\mbox{\cite{villalongaFlexibleHighperformanceSimulator2019, pednaultLeveragingSecondaryStorage2019}},
 quasiprobability/stabilizer-rank methods~\cite{seddonQuantifyingQuantumSpeedups2020} (and references therein), as well as decision diagrams~\cite{niemannQMDDsEfficientQuantum2016, zulehnerAdvancedSimulationQuantum2019,zulehnermatrix-vector2019}.

Second, as in the verification of classical circuits, the quality of the verification process heavily depends on the applied set of stimuli, i.e., 100\% certainty cannot be guaranteed as long as the set of applied stimuli is not exhaustive. 
Moreover, while the stimuli space for classical circuits is finite (each input bit can be assigned either $0$ or $1$---yielding a total of~$2^n$ possible stimuli), the state space in the quantum realm is infinitely large (possible stimuli are elements of a $2^n$-dimensional Hilbert space). 
This raises the question on whether simulative verification of quantum circuits (on classical computers) is suitable at all and, if so, how to generate proper stimuli~\ket{\varphi} to efficiently check the correctness of a quantum circuit.

In the following, we show that, although the perspective of a possible infinite number of stimuli may seem rather grim at a first glance, there are promising ways to check the correctness of quantum circuits using simulative verification. These, however, severely depend on how the stimuli are actually generated.
In fact, we show (both, theoretically and empirically) that high error detection rates can be achieved even if only a few randomly-chosen stimuli are considered---as long as these are generated in a specific fashion.

\section{Random Stimuli Generation}
\label{sec:proposed}
\vspace*{-0.35em}

In this section, we propose different schemes for the generation of (random) stimuli and explore how well they can show the correctness 
of a quantum circuit. 
To this end, each of the following subsections introduces, illustrates, and (theoretically) analyzes different stimuli generation schemes and their likeliness of detecting errors. 
Eventually, this will show that simulative verification indeed is very promising since sets of stimuli can be generated in a fashion
that offers a nice trade-off between error detection rate (as well as the required number of stimuli) and efficiency. 
The results of these conceptual and theoretical considerations have also been empirically confirmed as summarized later in Section~\ref{sec:results}.

\subsection{Classical Stimuli}
\label{sec:computationalbasis}
\vspace*{-0.35em}

The most straight-forward application of simulative verification for quantum circuits (compared to the classical approach reviewed in Section~\ref{sec:classical}) is to consider the set of 
computational basis states as stimuli (i.e.,~picking \ket{\varphi} from the set $\{\ket{i}\colon\;i\in\{0,1\}^n\}$) and computing $\mathcal{F}(U\ket{i}, V\ket{i})$, where $V$ is the matrix associated to $G$).
This has recently been studied in~\cite{burgholzerPowerSimulationEquivalence2020}, where empirical results show that choosing \enquote{classical} stimuli from this set at random \emph{often} allows to detect even small errors in quantum circuits. The following example illustrates this remarkable \enquote{power of simulation}.

\begin{example}\label{ex:single-qubit-errors}
Consider a certain $n$-qubit unitary specification $U$ and assume that some error affects (w.l.o.g.) the first qubit in the actual realization $G$. 
In the quantum realm, this means that the circuit $G$ is described by the unitary matrix $V = U \cdot (\mathbb{I}^{\otimes (n-1)} \otimes E)$, where $E$ describes an error gate that is applied to the first qubit.
Due to the inherent reversibility of quantum gates, this error has a localized effect on the output,~i.e., 
\begin{equation*}
\mathcal{F}(U|c \rangle, V |c \rangle) = \mathcal{F} (|c \rangle, (\mathbb{I}^{\otimes (n-1)}\otimes E) |c \rangle ) = | \langle c_0 | E |c_0 \rangle |^2 
\end{equation*}
for any classical stimulus $|c \rangle = |c_{n-1}\ldots c_0 \rangle$. 

Now suppose that \mbox{$E=X$},~i.e.,~ a bit flip error occured. 
In contrast to classical intuition, such an error can be detected by a \emph{single} simulation with \emph{any} classical stimulus~\ket{c}, since \mbox{$\mathcal{F}(U\ket{c}, V\ket{c}) = | \langle c_0 |X |c_0 \rangle|^2=0$} independent of~$\ket{c}$. 
\end{example}

However, this approach has a severe handicap which has not been discussed so far---namely that it is not faithful. 
Specifically, for each unitary specification $U$ there is an (infinitely large) family of realizations $G$ for which $\mathcal{F}(U\ket{c}, V\ket{c}) = 1$ holds for all classical stimuli~\ket{c}, even if %
quantum states \ket{\varphi} with %
$\mathcal{F}(U\ket{\varphi}, V\ket{\varphi}) \neq 1$ actually exist.
 An example illustrates the problem:
\begin{example}\label{ex:single-qubit-errors-handicap}
	Consider the same scenario as in Ex.~\ref{ex:single-qubit-errors}, but assume that the error is characterized as $E=Z$,~i.e.,~a phase flip error occurred.
	No classical stimulus~\ket{c} may detect such an error due to the fact that 
	$\mathcal{F}(U|c\rangle, V|c \rangle ) = | \langle c_0 | Z |c_0 \rangle|^2 = 1$ independent of \ket{c}.
	Intuitively, this happens whenever the \enquote{difference} of $U$ and $V$ is diagonal in the computational basis, such as $\mathbb{I}^{\otimes(n-1)}\otimes Z$ in case of this example.
\end{example}

Nevertheless, our empirical evaluations (which are summarized later in Section~\ref{sec:results}) show that whenever classical stimuli are actually capable of detecting a certain error in the realization $G$, they do so within remarkably few simulations with randomly picked classical stimuli---an effect contradictory to classical intuition as already observed in~\cite{burgholzerPowerSimulationEquivalence2020}. 

\vspace*{-2mm}
\subsection{Local Quantum Stimuli}
\label{sec:randombasis}
\vspace*{-0.35em}

In the previous section, we showed that classical stimuli generation is not sufficient to faithfully detect errors in quantum circuits. On an abstract level, this should not come as a surprise. After all, quantum circuits are designed to achieve tasks that classical circuits cannot.
In fact, a closer look at the single-(qu)bit case already reveals a fundamental discrepancy: 
Classical single-bit operations map one of two possible inputs ($0$ or $1$) to one of two possible outputs ($0$ or~$1$). In contrast, the quantum case is much more expressive: %
The set of all possible single-qubit states $\ket{\varphi}$ is infinitely large and can be parametrized by the 2-dimensional Bloch sphere \cite{nielsenQuantumComputationQuantum2010} illustrated in Figure~\ref{fig:bloch-sphere}. Single-qubit quantum operations map single-qubit states to single-qubit states. Geometrically, this family encompasses all possible rotations of the Bloch sphere as well as all reflections. 
Classical (single-qubit) stimuli,~i.e.,~the states \ket{0} and \ket{1}, are not expressive enough to reliably probe such a continuum of operations. They correspond to antipodal points on the (Bloch) sphere and it is simply impossible to detect certain transformations by tracking the movement of only two antipodal points.

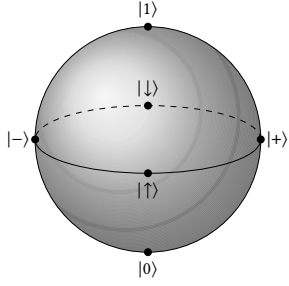
\begin{figure}
\centering
\scalebox{0.75}{
		\begin{tikzpicture}
		\shade[ball color = gray!40, opacity = 0.4] (0,0) circle (2cm);
		\draw (0,0) circle (2cm);
		\draw (-2,0) arc (180:360:2 and 0.6);
		\draw[dashed] (2,0) arc (0:180:2 and 0.6);
		\fill[fill=black] (2,0) circle(2pt);
		\node at (2.3,0) {$|+ \rangle$};
		\fill[fill=black] (-2,0) circle(2pt);
		\node at (-2.3,0) {$|- \rangle $};
		\fill[fill=black] (0,2) circle(2pt);
		\node at (0,2.3) {$|1 \rangle$};
		\fill[fill=black] (0,-2) circle(2pt);
		\node at (0,-2.3) {$|0\rangle$};
		\fill[fill=black] (0,-0.6) circle(2pt);
		\node at (0,-0.9) {$|\!\uparrow \rangle$};
		\fill[fill=black] (0,0.6) circle(2pt);
		\node at (0,0.9) {$|\!\downarrow \rangle $};
		\end{tikzpicture}}\vspace*{-2mm}
	\caption{Bloch Sphere%
	}\vspace*{-5mm}
	\label{fig:bloch-sphere}
\end{figure}

In order to address this, also stimuli beyond (classical) basis states should be considered. More precisely, three pairs of antipodal points are sufficient for full resolution~\cite{schwingerUnitaryOperatorBases1960, klappeneckerMutuallyUnbiasedBases2005,kuengQubitStabilizerStates2015}, namely %
\begin{align*}
&\ket{0},& &\ket{1},& &\text{($Z$-basis)}, \\
&\ket{+}= \nicefrac{1}{\sqrt{2}}(\ket{0}+\ket{1}),& &\ket{-} = \nicefrac{1}{\sqrt{2}}(\ket{0}-\ket{1}),& &\text{($X$-basis)}, \mbox{and}\\
&\ket{\uparrow}= \nicefrac{1}{\sqrt{2}}(\ket{0}+i\ket{1}),& &\ket{\downarrow} = \nicefrac{1}{\sqrt{2}}(\ket{0}-i\ket{1}),& &\text{($Y$-basis)}.
\end{align*}
Generating stimuli uniformly at random from this sextuple\footnote{The single-qubit states $\ket{0},\ket{1},\ket{+},\ket{-},\ket{\uparrow},\ket{\downarrow}$ can be generated from the basis state \ket{0} by applying the gates $\mathbb{I}$, \qop{X}, \qop{H}, \qop{XH}, \qop{HS}, or \qop{XHS}, respectively.} produces a set that is expressive enough to detect \emph{any} \mbox{single-qubit} error. More precisely, for any pair of functionally different \mbox{single-qubit} unitaries~$U$ and~$V$, at least one input \mbox{$\ket{l_1} \in \left\{ \ket{0},\ket{1},\ket{+},\ket{-},\ket{\uparrow},\ket{\downarrow}\right\}$}
produces functionally different outputs, i.e., the fidelity~$\mathcal{F}(U\ket{l_1}, V \ket{l_1})$  is guaranteed to be $\neq 1$.

This desirable feature extends to the multi-qubit case. That is, if we independently select one of these six (single-qubit) states for every available qubit, every \enquote{local} single-qubit error may be detected. Thus, for $n$ qubits, we consider the following ensemble of \emph{local quantum stimuli}: 
\begin{equation}
\ket{l} = \ket{l_{n-1}} \otimes \cdots \otimes \ket{l_0} \text{ with } \ket{l_i} \in \left\{ \ket{0},\ket{1},\ket{+},\ket{-},\ket{\uparrow},\ket{\downarrow}\right\}
\label{eq:local-stimuli}
\end{equation}

\begin{example}\label{ex:single-qubit-errors-local}
Let us revisit the scenario from Ex.~\ref{ex:single-qubit-errors} (and Ex.~\ref{ex:single-qubit-errors-handicap}).
Compared to classical stimuli, local quantum stimuli behave in a more homogeneous fashion on the classical extreme cases shown before:
First, suppose that $E=X$ (bit flip error). Then, 
	\[ \mathcal{F}(U|l \rangle, V|l \rangle) = | \langle l_0| X |l_0 \rangle|^2 =\left\{\begin{array}{ll}0 & |l_0 \rangle \in \left\{|0 \rangle, |1 \rangle, |\!\uparrow, |\!\downarrow \rangle \right\} \\ 1 & |l_0 \rangle \in \left\{|+ \rangle, |- \rangle \right\} \end{array}\right. \] 
Compared to classical stimuli, only $\nicefrac{2}{3}$ of all local quantum stimuli detect this type of error.
Now, suppose that $E=Z$ (phase flip error). Then, 
\[ \mathcal{F}(U|l \rangle, V|l \rangle) = | \langle l_0| Z |l_0 \rangle|^2 =\left\{\begin{array}{ll}0 & |l_0 \rangle \in \left\{|+ \rangle, |- \rangle, |\!\uparrow, |\!\downarrow \rangle \right\} \\ 1 & |l_0 \rangle \in \left\{|0 \rangle, |1 \rangle \right\} \end{array}\right. \]
Consequently, in contrast to not detecting such an error with classical stimuli at all, again $\nicefrac{2}{3}$ of all local quantum stimuli are capable of detecting this type of error.
\end{example}

This observation that local quantum stimuli can detect errors which would have remained undetected using classical stimuli is not a coincidence. 
In fact, the collection of a total of $6^n$ local quantum stimuli is expressive enough to detect \emph{any} error in a quantum circuit.

\begin{theorem}\label{thm:faithful} For each pair of functionally distinct $n$-qubit unitaries~$U$ and~$V$, there exists at least one local quantum stimulus \ket{l} as defined in Eq.~\eqref{eq:local-stimuli}
	that detects the error, i.e., yields~$\mathcal{F}(U \ket{l}, V \ket{l} ) \neq1$.
\end{theorem}

\begin{proof}[Proof sketch\protect\footnotemark] \footnotetext{Note that, due to page limitations, we only provide a sketch of the proof for this theorem.}
	The key idea is to relate the expected fidelity~$\mathbb{E}_{\ket{l}} \mathcal{F}(U \ket{l}, V \ket{l})$---where the average is taken over all $6^n$ locally random stimuli---to a meaningful distance measure in the space of unitary matrices. This average outcome fidelity equals~1 if and only if $U$ and $V$ are functionally equivalent. 
	Now, suppose that $U$ and $V$ are functionally distinct unitaries. Then,~$\mathbb{E}_{\ket{l}} \mathcal{F}(U \ket{l}, V \ket{l}) <1$ which is only possible if (at least) one stimulus $\ket{l}$ produces an outcome fidelity that is strictly smaller than one.
\end{proof}

While this rigorous statement asserts that any error can be detected by (at least) one local quantum stimulus, it does not provide any advice on how to find the ``right'' stimulus. This is a very challenging problem in general, but the above example suggests that repeated random sampling of stimuli should \enquote{do the job}. Our empirical studies (see Section~\ref{sec:results}) confirm that such a procedure works remarkably well. Typically, few randomly generated local quantum stimuli suffice to detect realistic errors. 

\vspace*{-0.35em}
\subsection{Global Quantum Stimuli}
\label{sec:randomclifford}
\vspace*{-0.35em}

The previous section has shown that a modest increase in the expressiveness of stimuli can already make a large difference. 
Local quantum stimuli can detect \emph{any} error, while classical stimuli cannot. 
This is interesting, because local quantum stimuli are comparatively few in number ($6^n$ states in a $2^n$-dimensional state space to detect arbitrary discrepancies in unitary circuits) and actually do not inherit many further quantum features.
For example, \enquote{global} quantum features such as entanglement are not employed by them at all.
This begs the question: what kind of advantages can even more expressive and ``more quantum'' stimuli offer?
Faithfulness is not a problem anymore, but richer, global stimuli may help to detect errors \emph{earlier},~i.e.,~after substantially fewer iterations. 

In order to identify powerful global quantum stimuli, it is helpful to revisit local quantum stimuli as introduced in Eq.~\eqref{eq:local-stimuli} from a different perspective:
They are generated through starting with a very simple classical state (i.e.,~\ket{0\ldots 0}) and applying certain \mbox{single-qubit} gates to the individual qubits,~e.g.,~$\ket{0} \otimes \ket{+} \otimes \ket{\uparrow} = (\mathbb{I} \otimes H \otimes HS) \ket{000}$.
 Consequently, \emph{random} local stimuli are generated by choosing this \emph{layer} of single-qubit gates at random. 
This generation scheme can be readily generalized.
Rather than selecting only a single layer of (single-qubit) gates, we construct a generation circuit~$G_{0} \cdots G_{l-1}$ 
that has $l > 1$ layers and, most importantly, also features two-qubit gates. That is, a stimuli $\ket{g}$ with $\ket{g} = (G_{0} \cdots G_{l-1}) \ket{0\ldots 0}$ is generated, where each~$G_i$ is a (single) layer comprised of so-called Clifford gates ($H$, $S$, $\qop{CNOT}$)~\cite{gottesmanStabilizerCodesQuantum1997}. 

Overall, this set of \emph{global quantum stimuli} $\ket{g}$ contains all local quantum stimuli, but is much richer and much more expressive. For instance, the overwhelming majority of global quantum stimuli will be highly entangled.
Provided that the number of layers $l$ is proportional to the number of qubits $n$~\cite{hunter-jonesUnitaryDesignsStatistical2019,brandaoLocalRandomQuantum2016}, these stimuli show remarkable properties. 
Most notably, the expected outcome fidelity (averaged over all possible global quantum stimuli $|g \rangle$) accurately approximates one of the most prominent distance measures for $n$-qubit quantum circuits, namely
\begin{equation}
\mathbb{E}_{|g \rangle} \mathcal{F}(U|g \rangle, V|g \rangle) \approx \mathcal{F}_{\mathrm{avg}}(U,V) = \tfrac{1}{2^n+1} \left( 1+ 2^n \big| \mathrm{tr}(U^\dagger V) \big|^2 \right).
\label{eq:average-fidelity}
\end{equation}
Here, $\mathrm{tr} (U^\dagger V)$ denotes the trace of the unitary matrix $U^\dagger V$.
This \emph{average (gate) fidelity}~\cite{nielsenQuantumComputationQuantum2010} forms the basis of many \mbox{state-of-the-art} quantum calibration procedures~\cite{magesanCharacterizingQuantumGates2012,kuengComparingExperimentsFaulttolerance2016}. 
Importantly, most (realistic) errors lead to an average fidelity that is tiny. Eq.~\eqref{eq:average-fidelity} allows us to capitalize on this phenomenon. The following statement is an immediate consequence of Eq.~\eqref{eq:average-fidelity} and Markov's inequality:

\begin{corollary}\label{thm:immediate-detection}
Consider a unitary specification $U$ and a particular realization as a quantum circuit $G$ (represented by the unitary $V$). 
Then, a randomly selected global quantum stimulus obeys
\begin{equation*}
\mathrm{Pr}_{|g \rangle} \left[ \mathcal{F}(U|g \rangle, V |g \rangle) =1 \right] \leq \mathcal{F}_{\mathrm{avg}}(U,V).
\end{equation*}
The r.h.s.\ equals~1 if and only if $G$ correctly realizes $U$, otherwise it is typically \emph{much} smaller.
\end{corollary}

This general %
statement does have powerful implications when applied to a precise example.

\begin{example}
Consider again the scenario from Ex.~\ref{ex:single-qubit-errors} (and Ex.~\ref{ex:single-qubit-errors-handicap}): A single-qubit error $E$ occurred on the first qubit leading to the unitary $V=U\cdot (\mathbb{I}^{\otimes (n-1)}\otimes E)$, where the single-qubit error is either $E=X$ (bit flip error) or $E=Z$ (phase flip error). 
Then, $\mathcal{F}_{\mathrm{avg}}(U,V) = \tfrac{1}{2^n+1}\leq 2^{-n}$ (because Pauli matrices are traceless) 
and Corollary~\ref{thm:immediate-detection} implies that it is very unlikely to \emph{not} detect this error with a single, random global quantum stimulus, i.e., $\mathrm{Pr}_{|g \rangle} \left[ \mathcal{F}(U|g \rangle, V |g \rangle)=1\right] \leq 2^{-n} \ll 1$. 
\end{example}

This example demonstrates the power of global quantum stimuli. However, it is important to keep in mind that this power is not for free. The generation of (random) global quantum stimuli and subsequent simulation is much more resource-intensive by comparison (as confirmed by our empirical evaluations in Section~\ref{sec:results}).

This can also be understood from a broader context:
The average (gate) fidelity as given by Eq.~\eqref{eq:average-fidelity} is closely related to another popular distance measure---the \emph{entanglement fidelity}. 
This quantity captures the performance of a powerful quantum stimulus $|\Omega \rangle$, see e.g.~\cite{khatriQuantumassistedQuantumCompiling2019}.
This stimulus is generated from $2n$ qubits by pairwise entangling individual qubits of one half of the system with the qubits of the other half. 
Applying both circuits to the first half of this state and computing the fidelity of the outcome states subsequently yields the entanglement fidelity~\mbox{\cite{schumacherSendingEntanglementNoisy1996,kitaevQuantumComputationsAlgorithms1997}},~i.e.,
\begin{equation}
\mathcal{F}(U \otimes \mathbb{I} |\Omega \rangle, V \otimes \mathbb{I} |\Omega \rangle ) 
= 4^{-n} \big| \mathrm{tr}(U^\dagger V ) \big|^2 = \mathcal{F}_{\mathrm{ent}}(U,V). 
\label{eq:entanglement-fidelity}
\end{equation}
Comparing Eq.~\eqref{eq:average-fidelity} and Eq.~\eqref{eq:entanglement-fidelity} shows that these quantities are almost identical.
This implies that global quantum stimuli accurately approximate the powerful quantum stimulus \ket{\Omega} on average.
Finally, we point out that conducting simulative verification with \ket{\Omega} itself is not feasible on classical computers, since requiring double the amount of qubits exponentially increases the resource-demand for classical simulations.

\begin{table*}[htb]
	\sisetup{table-number-alignment=center, table-text-alignment=center, separate-uncertainty, detect-weight, detect-inline-weight=math, table-figures-uncertainty=1}
	\centering
	\caption{Experimental results (quantities averaged over a total of approx.~$10^6$ different simulations)}
	\label{tab:results}\vspace*{-1em}
	\footnotesize
	\resizebox{0.96\linewidth}{!}{
	\begin{tabular}{@{}l*{9}{!{\quad}S[round-integer-to-decimal]}@{}}\toprule
		 & \multicolumn{3}{c}{Remove 1 random gate} & \multicolumn{3}{c}{Remove 2 random gates} & \multicolumn{3}{c}{Remove 3 random gates} \\
		\cmidrule(r){2-4} \cmidrule(lr){5-7} \cmidrule(l){8-10} 
		{Approach} & {$p_{\mathit{s}}$ [\si{\percent}]} & {$\varnothing\mathit{s}$} & {$\varnothing\mathit{t}$ [\si{\second}]} & {$p_{\mathit{s}}$ [\si{\percent}]} & {$\varnothing\mathit{s}$} & {$\varnothing\mathit{t}$ [\si{\second}]} & {$p_{\mathit{s}}$ [\si{\percent}]} & {$\varnothing\mathit{s}$} & {$\varnothing\mathit{t}$ [\si{\second}]}
		\\\midrule
		\begin{filecontents*}{./csv/evaluation_with_uncertainty.csv}
Classical Stimuli (Section~\ref{sec:computationalbasis}),86.9 \pm 4.1,4.0 \pm 0.8,\bfseries 0.2 \pm 0.3,97.9 \pm 2.2,1.7 \pm 0.4,\bfseries 0.1 \pm 0.1,99.6 \pm 0.8,1.2 \pm 0.2,\bfseries 0.1 \pm 0.1,54.9 \pm 4.7,7.8 \pm 0.7,\bfseries 0.4 \pm 0.5,80.7 \pm 5.0,3.9 \pm 0.8,\bfseries 0.2 \pm 0.2,90.8 \pm 4.2,2.4 \pm 0.6,\bfseries 0.1 \pm 0.1,82.0 \pm 11.7,5.3 \pm 1.9,\bfseries 0.5 \pm 0.7,80.3 \pm 12.2,5.3 \pm 2.0,\bfseries 0.5  \pm 0.7,84.1 \pm 5.6,3.9 \pm 0.9,\bfseries 0.3 \pm 0.3
Local Quantum Stimuli (Section~\ref{sec:randombasis}),98.8 \pm 1.6,1.5 \pm 0.3,0.7 \pm 1.3,\bfseries 100.0 \pm 0.0,1.1 \pm 0.1,0.5 \pm 0.9,\bfseries 100.0 \pm 0.0,\bfseries 1.0 \pm 0.0,0.7 \pm 1.1,73.9 \pm 7.4,5.1 \pm 1.1,2.9 \pm 5.1,92.5 \pm 3.9,2.2 \pm 0.6,1.1 \pm 2.0,97.5 \pm 2.8,\bfseries 1.4 \pm 0.4,0.6 \pm 0.9,82.3 \pm 11.6,4.0 \pm 1.8,2.8 \pm 5.3,80.6 \pm 12.0,4.1 \pm 1.8,2.5 \pm 5.1,90.7 \pm 4.9,2.5 \pm 0.8,1.5 \pm 2.7
Global Quantum Stimuli (Section~\ref{sec:randomclifford}),\bfseries 99.0 \pm 1.5,\bfseries 1.2 \pm 0.2,50.1 \pm 103.0,\bfseries 100.0 \pm 0.0,\bfseries 1.0 \pm 0.0,56.9 \pm 113.1,\bfseries 100.0 \pm 0.0,\bfseries 1.0 \pm 0.0,68.7 \pm 115.3,\bfseries 75.9 \pm 10.1,\bfseries 4.6 \pm 1.5,80.9 \pm 118.1,\bfseries 92.9 \pm 4.3,\bfseries 2.1 \pm 0.6,47.9 \pm 97.1,\bfseries 97.6 \pm 2.8,\bfseries 1.4 \pm 0.4,38.2 \pm 93.1,\bfseries 82.9 \pm 12.1,\bfseries 3.6 \pm 1.8,79.9 \pm 120.2,\bfseries 81.2 \pm 12.6,\bfseries 3.8 \pm 1.9,66.7 \pm 116.7,\bfseries 91.2 \pm 5.4, \bfseries 2.3 \pm 0.8, 61.2 \pm 109.6
\end{filecontents*}
\csvreader[column count=28, no head]{./csv/evaluation_with_uncertainty.csv}
		{1=\approach, 
		2=\remonep, 3=\remones, 4=\remonet, 
		5=\remtwop, 6=\remtwos, 7=\remtwot, 
		8=\remthreep, 9=\remthrees, 10=\remthreet, 
		11=\addonep, 12=\addones, 13=\addonet,
		14=\addtwop, 15=\addtwos, 16=\addtwot,
		17=\addthreep, 18=\addthrees, 19=\addthreet,
		20=\toffBp, 21=\toffBs, 22=\toffBt,
		23=\toffEp, 24=\toffEs, 25=\toffEt,
		26=\avgP, 27=\avgS, 28=\avgT}
		{\approach &
		\remonep & \remones & \remonet &
		\remtwop & \remtwos & \remtwot &
		\remthreep & \remthrees & \remthreet \cr }
		\\[-\normalbaselineskip]\toprule
		& \multicolumn{3}{c}{Add 1 random gate} & \multicolumn{3}{c}{Add 2 random gates} & \multicolumn{3}{c}{Add 3 random gates}\\
		\cmidrule(r){2-4} \cmidrule(lr){5-7} \cmidrule(l){8-10} 
		{Approach} & {$p_{\mathit{s}}$ [\si{\percent}]} & {$\varnothing\mathit{s}$} & {$\varnothing\mathit{t}$ [\si{\second}]} & {$p_{\mathit{s}}$ [\si{\percent}]} & {$\varnothing\mathit{s}$} & {$\varnothing\mathit{t}$ [\si{\second}]} & {$p_{\mathit{s}}$ [\si{\percent}]} & {$\varnothing\mathit{s}$} & {$\varnothing\mathit{t}$ [\si{\second}]}
		\\\midrule
		\csvreader[column count=28, no head]{./csv/evaluation_with_uncertainty.csv}
		{1=\approach, 
		2=\remonep, 3=\remones, 4=\remonet, 
		5=\remtwop, 6=\remtwos, 7=\remtwot, 
		8=\remthreep, 9=\remthrees, 10=\remthreet, 
		11=\addonep, 12=\addones, 13=\addonet,
		14=\addtwop, 15=\addtwos, 16=\addtwot,
		17=\addthreep, 18=\addthrees, 19=\addthreet,
		20=\toffBp, 21=\toffBs, 22=\toffBt,
		23=\toffEp, 24=\toffEs, 25=\toffEt,
		26=\avgP, 27=\avgS, 28=\avgT}
		{\approach &
		\addonep & \addones & \addonet &
		\addtwop & \addtwos & \addtwot &
		\addthreep & \addthrees & \addthreet \cr}
		\\[-\normalbaselineskip]\toprule
		& \multicolumn{3}{c}{Add 10 random Toffolis at beginning} & \multicolumn{3}{c}{Add 10 random Toffolis at end}& \multicolumn{3}{c}{\bfseries Overall}\\
		\cmidrule(r){2-4} \cmidrule(lr){5-7}  \cmidrule(l){8-10} 
		{Approach} & {$p_{\mathit{s}}$ [\si{\percent}]} & {$\varnothing\mathit{s}$} & {$\varnothing\mathit{t}$ [\si{\second}]} & {$p_{\mathit{s}}$ [\si{\percent}]} & {$\varnothing\mathit{s}$} & {$\varnothing\mathit{t}$ [\si{\second}]} & {$p_{\mathit{s}}$ [\si{\percent}]} & {$\varnothing\mathit{s}$} & {$\varnothing\mathit{t}$ [\si{\second}]} 
		\\\midrule
		\csvreader[column count=28, no head]{./csv/evaluation_with_uncertainty.csv}
		{1=\approach, 
		2=\remonep, 3=\remones, 4=\remonet, 
		5=\remtwop, 6=\remtwos, 7=\remtwot, 
		8=\remthreep, 9=\remthrees, 10=\remthreet, 
		11=\addonep, 12=\addones, 13=\addonet,
		14=\addtwop, 15=\addtwos, 16=\addtwot,
		17=\addthreep, 18=\addthrees, 19=\addthreet,
		20=\toffBp, 21=\toffBs, 22=\toffBt,
		23=\toffEp, 24=\toffEs, 25=\toffEt,
		26=\avgP, 27=\avgS, 28=\avgT}
		{\approach &
		\toffBp & \toffBs & \toffBt &
		\toffEp & \toffEs & \toffEt &
		\avgP & \avgS & \avgT \cr}
		\\[-\normalbaselineskip]\bottomrule
	\end{tabular}}\\
{\footnotesize {$p_{\mathit{s}}$ [\si{\percent}]}: Error detection rate in percent  \hspace*{0.45cm} {$\varnothing\mathit{s}$}: Average number of stimuli  \hspace*{0.45cm} {$\varnothing\mathit{t}$ [\si{\second}]}: Average runtime in seconds \\
Since the obtained results are rather homogeneous across the respective benchmarks (as confirmed by the moderate standard deviation), we only list averaged values here.}\vspace*{-1.5em}
\end{table*}

\vspace*{-0.55em}
\section{Empirical Study}
\label{sec:results}
\vspace*{-0.35em}

In this section, we empirically study the behavior of the schemes proposed in Section~\ref{sec:proposed} through extensive evaluations.
To this end, the proposed schemes have been implemented in $\CC$ as part of the open-source JKQ framework for quantum computing~\cite{willeJKQJKUTools2020}.
More precisely, they have been integrated into the JKQ~QCEC quantum circuit equivalence checking tool (publicly available at~\mbox{\url{https://github.com/iic-jku/qcec}}) using the decision diagram-based simulator from~\cite{zulehnerAdvancedSimulationQuantum2019} for conducting the simulations.
In order to obtain a rigorous evaluation, 
we considered the following setup:
\begin{itemize}\sloppy
	\item We chose $25$ widely-used reversible/quantum algorithms %
	with \SIrange[range-units=single]{16}{34}{qubits}---constituting the respective reference implementations~$U$.
	\item Each algorithm has been compiled to a suitable IBM architecture using IBM Qiskit~\cite{aleksandrowiczQiskitOpensourceFramework2019}---constituting the realization~$G$.
	\item In order to study the detection of errors, a total of $8$ \mbox{error-injection} options have been considered for each circuit\footnote{In any realistic scenario where, e.g., a bug is present in the compilation flow, the resulting errors in $G$ would be much more severe than the error-injections studied in this work. Consequently, it can be deducted from the results obtained in this work that the proposed schemes perform even more reliably on such instances.}:
	\begin{itemize}
		\item Removal of $1$, $2$, or $3$ random gates from $G$,
		\item Insertion of $1$, $2$, or $3$ random gates from the set $\{X,Y,Z,H,S,T\}$ on random qubits into $G$,
		\item Insertion of $10$ random Toffoli gates at the beginning or at the end of $G$.
	\end{itemize}
	\item For each error-injection option, $50$ random seeds have been considered.
	\item For each resulting instance, $5$ random seeds have been used for randomly picking stimuli according to the respective scheme.
	\item For each resulting instance and random seed, up to $16$ simulations of $U$ and $G$ with stimuli randomly picked according to the specific scheme have been performed aiming to detect the injected error.
\end{itemize}
Overall, this led to a total of \si{50\,000} benchmark instances. Since for each instance on average approximately $3$ random stimuli were necessary to detect the error, a grand total of approx.~$10^6$ simulations have been conducted. 
To the best of our knowledge, this led to the broadest empirical evaluation of
simulative verification schemes for quantum circuits
to date.

The obtained results are summarized in Table~\ref{tab:results}. 
Here, we list the error detection rate $p_{\mathit{s}}$ in percent (i.e., the probability that the error is detected by the generated set of stimuli), the number of stimuli $\varnothing\mathit{s}$ needed to detect the error, and the runtime $\varnothing\mathit{t}$ of the respective scheme in seconds\footnote{The runtime depends on the simulator used, as well as the hardware the simulations are conducted on. Nevertheless, it allows to reason about the efficiency of the individual schemes to some extent.}.
Due to page limitations, we only list the averaged values (w.r.t.~the different error injections). However, since the obtained results are rather homogeneous across the respective benchmarks (as confirmed by the moderate standard deviation which is also listed in Table~\ref{tab:results}, this still allows for a proper interpretation of the results. 

\vspace{200cm}
From those results, the following conclusions can be drawn:
\begin{itemize}
	\item \emph{All} schemes lead to sets of stimuli with remarkable error detection rates. With randomly chosen stimuli only, few stimuli are sufficient to detect the vast majority of errors (while, in contrast, dedicated constrained-based stimuli generation, fuzzing, etc. methods~\cite{yuanConstraintbasedVerification2006,bergeronWritingTestbenchesUsing2006,kitchenStimulusGenerationConstrained2007,wille2009smt, leDetectionHardwareTrojans2019,laeuferRFUZZCoveragedirectedFuzz2018} are required in the classical realm to get a merely acceptable error detection rate).
	
	\item Based on these high standards, classical stimuli generation performs worst and often fails---especially in cases where individual (diagonal) gates are removed or added. This is a consequence of classical stimuli not being faithful as shown in Section~\ref{sec:computationalbasis}. At the same time, the corresponding simulations are very fast; making this scheme suitable for rapid prototyping.
	
	\item On the other side of the spectrum, global stimuli generation yields the most robust results,~i.e.,~requiring the least amount of stimuli and also achieving the highest error detection rates. This confirms the discussions from Section~\ref{sec:randomclifford} on the quality of those stimuli. Thus, this scheme is suitable for rigorous testing even if the simulation of those stimuli is severely more runtime-demanding. 
	
	\item Local quantum stimuli generation constitutes a trade-off  
	between quality and efficiency compared to the other two schemes. 
	Although this scheme is not as powerful as global quantum stimuli generation with respect to quality, it is faithful (as shown in Section~\ref{sec:randombasis}) and remains rather efficient. 
	
\end{itemize}

\vspace*{-0.35em}
\section{Conclusion}
\label{sec:conclusions}
\vspace*{-0.35em}

In this work, we showed that simulative verification in the quantum realm is much more powerful than in the classical realm.
On the one hand, we introduced, illustrated, and analyzed three potential quantum stimuli generation schemes offering a trade-off between error detection rate (as well as the required number of stimuli) and efficiency.
On the other hand, we showed (both, theoretically and empirically) that, in contrast to classical circuits, high error detection rates can be achieved by just considering a few \mbox{randomly-chosen} stimuli (generated according to the proposed schemes).
This eventually shows  that simulative verification offers huge potential in the verification of quantum circuits. 

\vspace*{-0.35em}
\begin{acks}
This work has partially been supported by the LIT Secure and Correct Systems Lab funded by the State of Upper Austria as well as by the BMK, BMDW, and the State of Upper Austria in the frame of the COMET program (managed by the FFG).
\end{acks}

\printbibliography
\end{document}